\documentclass{article}
\usepackage[american]{babel}
\usepackage{natbib} 
\usepackage{amsfonts,dsfont,bm,amsthm}
\usepackage{amsmath, amssymb, graphicx, url, algorithm2e}
\usepackage{mathtools, subfigure} 
\usepackage{booktabs} 
\usepackage{multirow} 
\usepackage{wrapfig}
\usepackage{thmtools, thm-restate}
\usepackage{amsfonts}
\usepackage{tikz} 
 \usepackage{pgfplots}
\usetikzlibrary{calc, trees, positioning, arrows, fit, shapes}
\usepackage{thm-restate} 
\usepackage{tikz}
\usetikzlibrary{arrows.meta}
\usepackage{subcaption}
\usepackage{booktabs}
\usepackage{caption}
\PassOptionsToPackage{hyphens}{url}
 \usepackage{hyperref}
\usepackage{multicol}
\usepackage{fullpage}
\usepackage{authblk}

 \newtheorem{assumption}{Assumption}
 \newtheorem{example}{Example}
 \newtheorem{definition}{Definition}
\newtheorem{lemma}{Lemma}

\definecolor{forestgreen}{RGB}{34,139,34}     
\definecolor{softgreen}{RGB}{144,238,144} 
\definecolor{darkgreen}{RGB}{0,100,0}     

\title{ Constraint-based difference graph discovery
in a linear setting}
\author[1]{\href{mailto:<daria.bystrova@univ-grenoble-alpes.fr>?Subject=Your UAI 2024 paper}{Daria Bystrova}{}}
\author[1]{Emilie Devijver}
\affil[1]{%
 Univ. Grenoble Alpes, CNRS, Grenoble INP, LIG, F38000, Grenoble, France\\
}
\date{}

\begin{document}

\maketitle

\begin{abstract}%
Comparing causal relationships across populations is essential in many scientific domains. This paper studies the problem of inferring a difference graph between two environments and proposes a causal discovery method for linear structural causal models based on equality tests of regression coefficients. 
We show that invariance of regression coefficients is governed by graphical conditions that go beyond standard d-separation. Therefore, we introduce \emph{diff-separation}, a graphical criterion that characterizes when a conditioning set blocks all paths capable of inducing differences in regression coefficients across environments.
Building on this criterion, 
we introduce a corresponding \emph{diff-faithfulness} assumption, linking graphical diff-separation statements to equality constraints on regression coefficients. Finally, we propose \textsc{LDiffPC}, a PC-style algorithm that uses equality tests of regression coefficients to recover the differences from multi-environment data.
\end{abstract}


\section{Introduction}

Understanding causal relationships among variables is essential in many domains, including medicine and ecology~\citep{belyaeva2021dci, pellissier2018comparing}. Beyond identifying causal structures, researchers often need to compare these relationships across different groups or environments to uncover meaningful variations. Such differences can have critical implications: in medicine, variations in causal mechanisms between patient subgroups may inform personalized treatments and improve patient outcomes; in ecology, shifts in causal interactions across ecosystems may indicate changes in resilience or vulnerability.

To make this concrete, consider a structural causal model (SCM) describing the relationship between a treatment $X$, a biomarker $Z$, and a health outcome $Y$, where $X$ affects $Z$ and $Z$ affects $Y$. Suppose that across two patient populations, the effect of $X$ on $Z$ changes due to differences in treatment response, while the rest of the system remains unchanged. In this case, the underlying causal structure is largely shared, but one mechanism differs. The goal is not to relearn the full causal graph in each environment, but rather to identify and characterize such differences. This motivates the notion of a \emph{difference graph}~\citep{wang2018direct, assaad2024causal}, which encodes changes in causal mechanisms across environments.

Recently, the problem of directly learning difference graphs from data has received increasing attention~\citep{wang2018direct, chen2024iscan, malik_identifying_2024}. However, existing methods are typically not formulated within the classical constraint-based framework of causal discovery, and lack a systematic graphical characterization of the conditions under which differences can be identified. In particular, there is currently no analogue of d-separation that explains when equality (or inequality) of statistical quantities across environments reflects underlying causal changes.

In this work, we address this gap by revisiting the problem through the lens of constraint-based causal discovery. We focus on linear SCMs and propose to use regression coefficients for testing the equality of direct effects across environments. We show that classical criteria based on conditional independence or mutual information are not well suited for this task, as they do not isolate changes in specific structural relationships.

Our main contribution is a graphical characterization of when regression coefficients remain invariant across environments. We show that this invariance depends on conditions that go beyond standard d-separation. In particular, paths that do not contain any changed edge may still induce differences. To capture these paths, we introduce a new graphical criterion, \emph{diff-separation}, which characterizes when a conditioning set blocks all paths that can induce differences in regression coefficients across environments.
Building on this notion, 
we  introduce a corresponding \emph{diff-faithfulness} assumption, linking graphical diff-separation to equality constraints on regression coefficients. Finally, we propose \textsc{LDiffPC}, a constraint-based algorithm that leverages equality tests of regression coefficients to recover the differences from multi-environment data.

The remainder of the paper is organized as follows. 
In Section~\ref{sec:sota}, we discuss related works.
In Section~\ref{sec:notions}, we review background on SCMs and causal discovery. In Section~\ref{sec:graph}, we introduce diff-separation and study its properties. 
In Section~\ref{sec:method}, we introduce the \textsc{LDiffPC} algorithm. Finally, Section~\ref{sec:conc} concludes the paper. All proofs are deferred to the appendix.

\section{Related work}\label{sec:sota}

When the causal graph is unknown, causal discovery methods aim to infer causal structures from observational data. One of the most well-known approaches is the PC algorithm \citep{Spirtes_2000} which assumes Faithfulness and Causal Sufficiency \citep{Spirtes_2000}, does not impose parametric assumptions on the underlying distribution.

Networks are widely used to represent relationships between components in complex systems. In many applications, researchers seek to understand how these dependency structures change across different environments. For instance, ecological studies have examined how species interaction networks evolve along elevation gradients~\citep{pellissier2018comparing, MontanoCentellas2020}. Similarly, in genomics, variations in gene expression networks between individuals who develop certain conditions and those who do not can provide critical insights into disease mechanisms~\citep{belyaeva2021dci}.

To address these challenges, difference network analysis methods have been developed for various types of graphical models. \cite{shojaie_differential_2021} provide a comprehensive review of existing approaches, which differ based on the type of network considered and the specific measure of difference employed. 
For Gaussian Graphical models, multiple methods developed for comparing covariance (or precision) matrices \citep{chen2015selection,chang2017comparing}.  \citep{guo_joint_2011} proposed joint modeling of graphical models corresponding to different environments by taking into account shared common structure. 
Alternatively, direct estimation of difference graphs has been proposed to infer structural changes without requiring the separate estimation of each graph. In the context of Markov networks, \cite{liu_direct_2013} and \cite{zhao_direct_2014} introduced methods that estimate differences in conditional independence structures directly from the data.

For Directed Acyclic Graphs (DAGs), \cite{wang2018direct} proposed the DCI algorithm, which identifies structural differences between two causal graphs by leveraging invariance tests. More recently, \cite{chen2024iscan} introduced iSCAN, a { semi-parametric approach that firstly infers the shifted nodes and then using a feature ordering method based on conditional independence recover  the difference DAG.} 
Both methods assume causal sufficiency, meaning that all relevant confounders are observed, and rely on parametric or {semi-parametric assumptions} about the underlying data distributions, such as linearity and Gaussianity or 
{additive noise model}. They further require that the two graphs share the same topological order, meaning that variables appear in a consistent causal sequence across environments. 
None of these methods fully adopts a PC-style strategy.

The problem of comparing difference graph is closely related to causal discovery from multiple environments. A substantial body of work has focused on identifying causal relationships that remain invariant across environments~\citep{wang2018direct}, exploiting distributional changes to recover stable causal structures. In contrast, we consider the complementary problem of identifying where and how causal mechanisms differ. While both perspectives rely on heterogeneity in the data, they address fundamentally different questions: invariant causal discovery seeks common structure, whereas difference-based approaches aim to characterize changes in mechanisms.

Recently, it has been demonstrated that difference graphs can be a valuable tool for causal reasoning~\citep{assaad2024causal}, enabling analysts to directly identify and estimate causal changes and effects from observational data. Importantly, this approach does not require parametric assumptions, making it a flexible and robust method for causal inference. However, when using difference graphs for causal estimation, it is crucial to avoid inheriting parametric assumptions from the causal discovery methods used to construct the difference graph, as this could compromise the assumption-free nature of the approach. 

\section{Preliminaries over Difference DAGs}\label{sec:notions}

\subsection{Underlying main SCM $\mathcal{M}^*$}
We consider an underlying \emph{structural causal model} (SCM)~\citep{Pearl_2000} denoted by $\mathcal{M}^*$. Formally, SCM is represented by a tuple $(\mathbb{U},\mathbb{V},\mathbb{F},\mathbb{P}^* )$ , where $\mathbb{U}$ is of a set of hidden (noise) variables  with an associated probability distribution, $\mathbb{V} = \{ X_1, \ldots, X_p \}$ is a set of observed random variables, and $\mathbb{F}$ is a set of causal mechanisms, where each $f_i \in \mathbb{F}$ generates $X_i$ from a subset of $\mathbb{V} \cup \mathbb{U}$. We assume that $\mathcal{M}^*$ induces a joint distribution $\mathbb{P}^*$ and a directed acyclic graph (DAG) $\mathcal{G}^* = (\mathbb{V}, \mathbb{E}^*)$, where the nodes in $\mathbb{V}$ represent observed variables, and the edges $\mathbb{E}^*$ represent causal relationships. We assume \emph{causal sufficiency}, i.e, all common causes of the observed variables are also included in the set. For any $X_i \in \mathbb{V}$, we denote by $Pa_{\mathcal{G}^*}(X_i)$ the set of parents of $X_i$ in $\mathcal{G}^*$.

An \emph{intervention} in $\mathcal{M}^*$, denoted $I = (X_i, \tilde{f}_i)$, replaces the causal mechanism $f_i$ of $X_i$ with a new mechanism $\tilde{f}_i$. This intervention modifies the conditional distribution $P(X|Pa_{G^*}(X))$ to a new distribution, $\tilde{P}(X|Pa_{G^*}(X))$. We assume that an intervention on $X_i$ can only modify the existing relationships between $X_i$ and its parents or reduce the size of $Pa_{\mathcal{G}^*}(X_i)$, which is formalized in the following assumption

\begin{assumption}[Shared topological ordering]
Let $\mathcal{M}^*$  $(\mathbb{U},\mathbb{V},\mathbb{F},\mathbb{P}^* )$ be an SCM. If $\mathbb{I}$ is a set of interventions in an SCM $\mathcal{M}^*$ and $\mathcal{M}^i$ is the corresponding SCM with the intervention distribution $\mathbb{P}^*$ and graph $G^{i}$, then for each $ I \in \mathbb{I}$, such $I = (X_i, \tilde{f}_i)$ and  $X_i = \tilde{f}_i(Pa_{\mathcal{G}^i}(X_i)$, where  $ Pa_{\mathcal{G}^i}(X_i) \subset Pa_{\mathcal{G}^*}(X_i)$
\end{assumption}

An \emph{environment} $e$ is defined as an SCM $\mathcal{M}^e$ that induces a DAG $\mathcal{G}^e = (\mathbb{V}, \mathbb{E}^e)$ and a joint distribution $\mathbb{P}^e$, which is determined by a set of interventions $\mathbb{I}^e$ applied to the original model $\mathcal{M}^*$.
We emphasize that our framework distinguishes between hard and soft interventions~\citep{Eberhardt_2007}. Under a hard intervention, the post-intervention distribution $\tilde{P}$ becomes degenerate at the imposed value, thereby severing the dependence between $X_i$ and all of its parents. By contrast, a soft intervention modifies the conditional mechanism of $X_i$ without necessarily removing all parental influences; it may affect only some parent-child relationships while leaving others unchanged.

Following the approach of \cite{Pearl_2000}, we assume that each causal mechanism $f_i$ is modular: when an intervention targets a node $X_j$, it only affects the mechanism $\mathbb{P}^e(X_j | \text{Pa}_{\mathcal{G}^*}(X_j))$ corresponding to $X_j$, while the mechanisms for other nodes, $\mathbb{P}^e(X_i | \text{Pa}_{\mathcal{G}^*}(X_i))$ for $i \neq j$, remain unchanged.

We also rely on standard concepts in DAGs. A path $\pi$ in a DAG is said to contain a \emph{V-structure} at node $X_j$ if the path includes $X_i \rightarrow X_j \leftarrow X_k$, and $X_j$ is called a \emph{collider}. The path $\pi$ is \emph{blocked} by a subset of vertices $\mathbb{Z}$ if a non-collider in $\pi$ belongs to $\mathbb{Z}$, or if $\pi$ contains a collider whose descendants do not belong to $\mathbb{Z}$. Otherwise, $\mathbb{Z}$ \emph{activates} $\pi$. If all paths between two nodes $X_i$ and $X_j$ are blocked by $\mathbb{Z}$, then $X_i$ and $X_j$ are \emph{d-separated} by $\mathbb{Z}$. 

\subsection{Difference DAG over two populations}
We consider that the underlying structural causal model $\mathcal{M}^*$, the causal graph $\mathcal{G}^*$, and the joint distribution $\mathbb{P}^*$ are unknown.
In this work, we assume there are exactly two environments, $e^1$ and $e^2$, with joint distributions $\mathbb{P}^1$ and $\mathbb{P}^2$, resulting from interventions on $\mathcal{M}^*$. We observe two datasets, $(\mathbf{x}^1_1, \dots, \mathbf{x}^1_{n_1})$ and $(\mathbf{x}^2_1, \dots, \mathbf{x}^2_{n_2})$, where each data point $\mathbf{x}^\ell_i \in \mathbb{R}^p$ (for $i \in \{1, \dots, n_\ell\}$ and $\ell \in \{1, 2\}$) is an observation derived from environment $e^\ell$.

As an intervention on $X_i$ cannot introduce new parents in $Pa_{\mathcal{G}^*}(X_i)$, any two graphs induced by any two environments of $\mathcal{M}^*$ share a common topological ordering, as previously utilized in \cite{wang2018direct, malik_identifying_2024}.  Consequently, no edge reversals are allowed between the two graphs, which is often a reasonable assumption in many practical applications \citep{belyaeva2021dci}. 

We aim to understand how causal relationships among variables differ between these two environments. To achieve this, we define the difference DAG~\citep{wang2018direct,assaad2024causal}, which captures causal changes between the observed environments. This graph, denoted as $\mathcal{D} = (\mathbb{V}, \mathbb{E})$, is formally defined as follows:

\begin{definition}[Difference DAG]\label{def:d_DAG} 
Let $\mathcal{M}^*$ be an SCM inducing a DAG $\mathcal{G}^*$. Consider two environments $e^1$ and $e^2$, where the respective SCMs $\mathcal{M}^1$ and $\mathcal{M}^2$ are obtained by applying interventions to $\mathcal{M}^*$.
A graph $\mathcal{D} = (\mathbb{V},\mathbb{E})$, with a set of vertices $\mathbb{V}=\{X_1,\ldots, X_p\}$, is called a \emph{difference DAG} of $\mathcal{M}_1$ and $\mathcal{M}_2$ if for each edge $X_i \rightarrow X_j$ in $\mathbb{E}$, $X_i\in Pa_{\mathcal{G}^*}(X_j)$ and the direct effect of $X_i$ on $X_j$ differs between $\mathcal{M}^1$ and $\mathcal{M}^2$.
\end{definition}

\begin{figure}[!t]
\centering
\begin{tabular}{p{5cm} p{3.4cm} p{5cm}}
\centering
\begin{tikzpicture}[
 xscale=0.55, >=Stealth,
  node/.style={circle, draw, minimum size=5mm, inner sep=0pt},
  every edge/.style={draw, ->},
    annotation/.style={
        font=\footnotesize,
        text=black!70,
        align=center
    }
]

\node[node] (A) at (0,6)  {X};

\node[node] (B) at (-2,5) {Y};
\node[node] (C) at ( 2,5) {Z};
\node[node] (D) at (0,4)  {W};

\node[node] (L1) at (-4,2) {$X$};
\node[node] (R1) at ( 4,2) {$X$};

\node[node] (L2) at (-5,1) {$Y$};
\node[node, fill=green!30] (L3) at (-3,1) {$Z$};
\node[node, ,  fill=red!30] (R2) at ( 3,1) {$Y$};
\node[node] (R3) at (5,1) {$Z$};

\node[node] (L4) at (-4,0) {$W$};
\node[node] (R4) at ( 4,0) {$W$};
\node[annotation] at (1,6.35)
    {$\mathcal{G}^{*}$};

\draw[->]  (A) -- (B);
\draw[->]  (A) -- (C);

\draw[->]  (B) -- (D);
\draw[->]  (C) -- (D);


\draw[->, thick] (-1.5,4)-- ++(-2,-1.5) node[midway, above] { $e^1$};
\draw[->, thick] (1.5,4)-- ++(2,-1.5) node[midway, above] { $e^2$};

\draw[->]  (L1) -- (L2);
\draw[->, green]  (L1) -- (L3);
\draw[->]  (L2) -- (L4);
\draw[->]  (L3) -- (L4);
\node[annotation] at (-3.3,-0.9)
    {$\mathcal{G}^1$};

\draw[->]  (R1) -- (R3);
\draw[->]  (R2) -- (R4);
\draw[->]  (R3) -- (R4);
\node[annotation] at (3.3,-0.9)
    {$\mathcal{G}^2$};

\end{tikzpicture}
&
$\begin{array}{c}
         \xrightarrow{\text{``Difference in causal relationships''}}  \\
     
    \end{array}$
&
\hspace*{1.2cm}%
\begin{tikzpicture}[{black, circle, draw, inner sep=0, font=\footnotesize, scale=1.2,
   annotation/.style={
        font=\footnotesize,
        text=black!70,
        align=center
    }}]
\tikzset{nodes={draw,rounded corners, thick},minimum height=0.6cm,minimum width=0.6cm}
		\node (X) at (0.75,2) {${X}$};
		\node (Y) at (0,0.67) {${Y}$} ;
		\node (Z) at (1.5,0.67) {${Z}$};
		\node (W) at (0.75,-0.66) {${W}$};
      \draw[->,>=latex] (X)  to (Z);
    	\draw[->,>=latex] (X) -- (Y);

\node[draw=none] at (0.75,-1.25) {$\mathcal{D}$};

    \end{tikzpicture}
\end{tabular}
\caption{Illustration of two environments derived from a common underlying DAG $\mathcal{G}^{*}$. A parametric (soft) intervention modifies the direct effect in one environment (green color), while a structural (hard) intervention alters the graph structure in the other (red color). The corresponding difference graph $\mathcal{D}$ is also shown, illustrating the object of interest by highlighting edges whose causal mechanisms differ across environments.}
\label{fig:dif-DAG}
\end{figure}
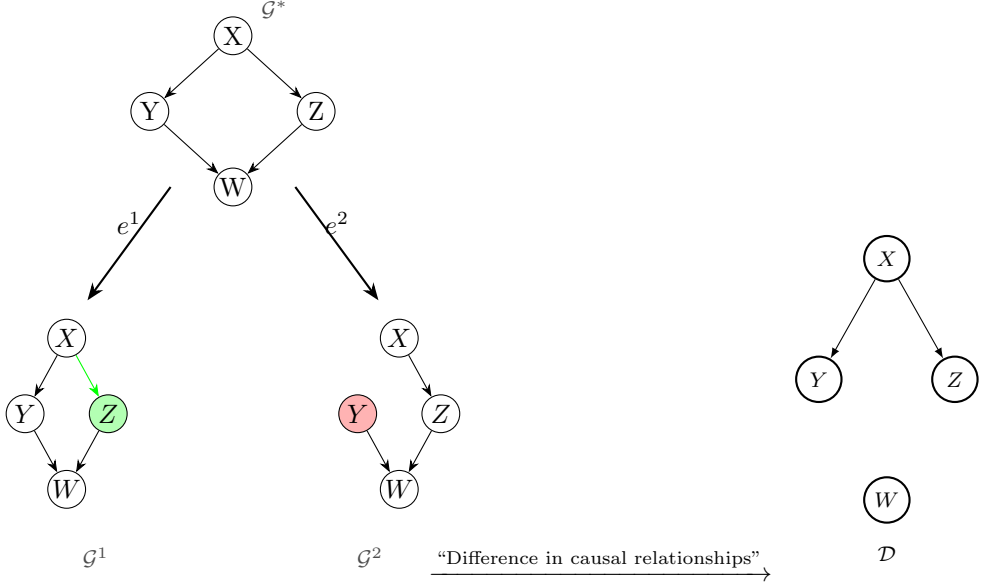

In a difference DAG, edge orientations are inherited from the underlying SCM $\mathcal{M}^*$, thereby preserving causal directionality across environments. Unlike in classical DAGs, the absence of an edge does not imply the absence of a causal relationship, but rather that the relationship is invariant across environments. Consequently, the difference DAG provides a compact representation of causal changes and serves as a basis for subsequent analysis.
Figure~\ref{fig:dif-DAG} illustrates a difference DAG $\mathcal{D}$. In this example, a soft intervention is applied to node $Z$ in environment $e1$, yielding the graph $G_1$, whereas a hard intervention is applied to $Y$ in environment  $e^2$ yielding the graph $G_2$.

\section{Difference DAG and (conditional) equalities} \label{sec:graph}

\subsection{Conditional equalities}

The goal of this paper is to develop an algorithm that detects the difference graph, or parts of it, from data using a strategy inspired by constraint-based causal discovery methods. Classical constraint-based methods rely on (conditional) independence tests derived from dependence measures to infer the underlying causal graph. A natural question is therefore whether similar dependence measures can be used, through equality tests across environments, to recover the difference graph.
For example, a common way to measure dependence between variables is through conditional mutual information \citep{thomas2006elements}.
One might think that changes  in dependence structure can be assessed by testing whether the conditional mutual information remains equal across the two environments:
$$I_{\mathbb{P}^1}(X_i, X_j | \mathbb{Z}) = I_{\mathbb{P}^2}(X_i, X_j | \mathbb{Z}),$$
where $I_{\mathbb{P}^1}$ represents the conditional mutual information  in the environment corresponding to an SCM $\mathcal{M}^{1}$ with probability distribution $\mathbb{P}^1$ and  $I_{\mathbb{P}^2}$ is the conditional mutual information in the environment associated with SCM $\mathcal{M}^{2}$ with probability distribution $\mathbb{P}^2$.  
However, the equality of the condition of mutual information does not necessarily indicate that the direct effect remains unchanged between two environments. 
Consider the causal structure
$X \rightarrow Y \leftarrow Z$,
where $X$ and $Z$ are independent and there are no unmeasured confounders. Suppose that the direct effect of $X$ on $Y$ remains unchanged across environments, while the effect of $Z$ on $Y$ varies. For instance, consider that in  $e^1$ we have
$
Y = a X + b Z + \epsilon_Y,
$
and in  $e^2$
$Y = a X + b' Z + \varepsilon_Y$,
such that $b' \neq b$ and suppose that all variables are Gaussian.
In this setting, the edge $X \rightarrow Y$ should not appear in the difference graph, since its corresponding causal coefficient remains invariant and equal to $a$ in both environments.
However, the marginal dependence between $X$ and $Y$, for instance measured by the mutual information $I(X;Y)$, may change across environments. For jointly Gaussian variables,
$$I(X;Y)
=
-\frac{1}{2}
\log\left(
1-
\frac{a^2Var(X)}
{a^2Var(X) + b^2Var(Z) + Var(\varepsilon)}
\right).
$$
Since this expression depends on \(b\), changing \(b\) to \(b'\) generally
changes the mutual information:
$I_{\mathbb{P}^1}(X;Y) \neq I_{\mathbb{P}^2}(X;Y)$.
The same problem also appear when usual simpler dependency measures such as (partial correlation). In the same example
$$\rho_{XY}
=
\frac{a\sqrt{Var(X)}}
{\sqrt{a^2Var(X)+b^2Var(Z)+Var(\varepsilon_Y)}}, 
$$
which depends on $b$, and may change across environments
even though the direct effect $a$ of $X$ on $Y$ is invariant.

To overcome this issue, it is necessary to directly directly assess changes in direct effects. In nonparametric settings, defining direct effects can be challenging: they may rely on interventional quantities such as controlled direct effects, which depend on the values of other variables, or on counterfactual definitions, which introduce additional complexity. To avoid these difficulties, we restrict our attention to linear SCMs, in which direct effects are uniquely defined by the structural coefficients and can be directly compared across environments.

\begin{assumption}
    The SCM $\mathcal{M}^*$ is assumed to be linear.
\end{assumption}
In a linear SCM, each causal mechanism is is defined by a coefficients that represent direct effects.
In the following, we will use an regression coefficient to asses where the direct effect changes between the environments: 
$$ {r_{XY.Z}}_{\mathbb{P}^1} ={r_{XY.Z}}_{\mathbb{P}^2}. $$

\subsection{Diff-separation}

We introduce diff-separation to establish a connection between measures of independence and changes in causal mechanisms. In the structural causal model (SCM) framework, d-separation is a fundamental tool used to assess conditional independence in graphs. Specifically, two nodes 
$X$ and $Y$ are d-separated if every path between them is blocked.
In the context of difference DAGs, however, the situation is different. A difference DAG contains only edges that represent changes in causal relationships across environments. Therefore, the absence of an active path between two nodes in the difference DAG does not imply that the nodes are independent, and classical d-separation is not directly applicable.
Instead, we rely on d-separation in the corresponding DAGs of each environment to evaluate dependencies reflected in the difference DAG. 

We illustrate why d-separation is not directly translatable to difference graphs using the following example. Suppose that the underlying SCM is $M^*$, with corresponding graph $G^*$ shown in Figure~\ref{fig:dif-sep}(a), namely $X \rightarrow Z \rightarrow Y$.
Assume that there is no intervention in environment  $e^2$, while in environment  $e^1$, there is only a structural intervention on $Y$. Then the difference graph $\mathcal{D}$ contains only the edge $Z \rightarrow Y$. In the corresponding difference DAG $D$, the nodes $X$ and $Y$ are d-separated given $\emptyset$.
We aim to relate the d-separation and the adjacency in the difference DAG, which is defined through the equality of the regression coefficients in the corresponding environments.  In this example,
$
{r_{XY.\emptyset}}_{\mathbb{P}^1} \neq {r_{XY.\emptyset}}_{\mathbb{P}^2}.
$
Indeed, in environment  $e^1$, we have
$
{r_{XY.\emptyset}}_{\mathbb{P}^1}=0,
$
whereas in environment  $e^2$, the regression coefficient becomes zero only after conditioning on $Z$:
$
{r_{XY.Z}}_{\mathbb{P}^2}=0.
$
Thus, rather than defining d-separation directly in the difference DAG, we should define it through d-separation in the corresponding environment-specific graphs.

A second issue is illustrated in Figure~\ref{fig:dif-sep}(b). We again start from the same SCM, but now consider two structural interventions: one on $Y$ in environment  $e^1$ and one on $Z$ in environment  $e^2$. In this case, the difference DAG contains both edges
$
X \rightarrow Z$ and
$Z \rightarrow Y.
$
If d-separation is defined only through the graphs $G_1$ and $G_2$, then $X$ and $Y$ are d-separated by the empty set. Moreover,
$
{r_{XY.\emptyset}}_{\mathbb{P}^1}
=
{r_{XY.\emptyset}}_{\mathbb{P}^2}.
$
However, this conflicts with the d-separation structure of the difference DAG $D$, which would instead suggest that $X$ and $Y$ are d-separated only after conditioning on $Z$.

Moreover, since we are interested in changes in dependence between nodes, we restrict attention to paths that involve interventions. Unlike classical DAGs, where all paths may be relevant, we focus specifically on those paths that include edges affected by interventions. Indeed, active paths that do not contain any edge affected by the intervention can be ignored, provided their associated path coefficients and relevant covariance structure remain unchanged, since they contribute identically before and after the intervention and therefore cancel out in the difference of regression coefficients.

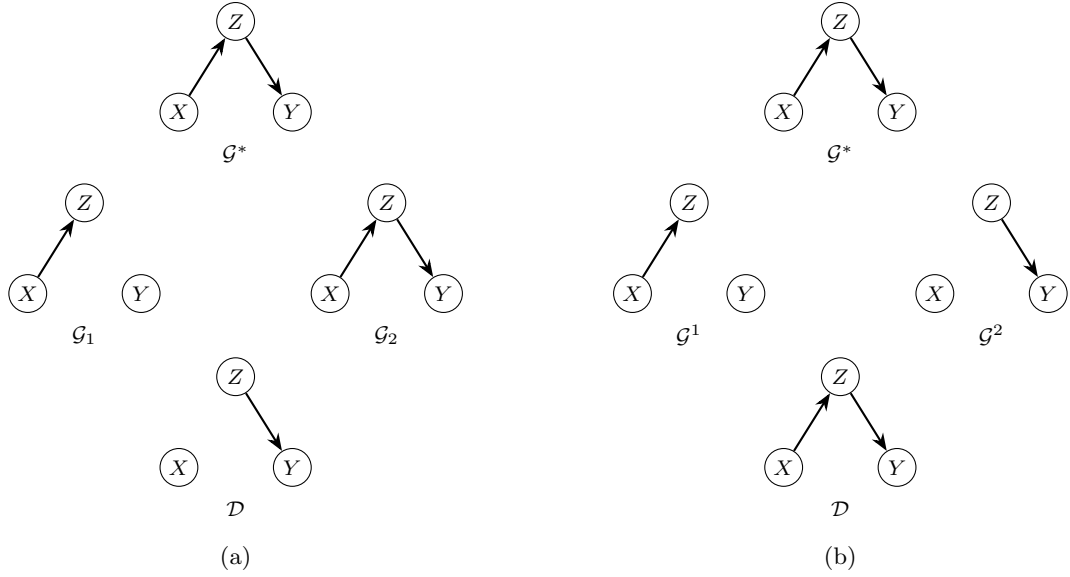
\begin{figure}[t]
\centering

\begin{tikzpicture}[
    x=1cm,
    y=1cm,
    >=Stealth,
    node/.style={
        circle,
        draw,
        minimum size=5mm,
        inner sep=0pt,
        font=\footnotesize
    },
    edge/.style={
        ->,
        thick
    },
    graphlabel/.style={
        draw=none,
        font=\footnotesize,
        inner sep=1pt
    },
    panellabel/.style={
        draw=none,
        font=\small
    }
]

\begin{scope}[shift={(0,0)}]

\node[node] (aXs) at (0,0)       {$X$};
\node[node] (aZs) at (0.75,1.2)  {$Z$};
\node[node] (aYs) at (1.5,0)     {$Y$};

\draw[edge] (aXs) -- (aZs);
\draw[edge] (aZs) -- (aYs);

\node[graphlabel] at (0.75,-0.55)
    {$\mathcal{G}^{*}$};

\node[node] (aX1) at (-2.0,-2.4)  {$X$};
\node[node] (aZ1) at (-1.25,-1.2) {$Z$};
\node[node] (aY1) at (-0.5,-2.4)  {$Y$};

\draw[edge] (aX1) -- (aZ1);

\node[graphlabel] at (-1.25,-2.95)
    {$\mathcal{G}_{1}$};

\node[node] (aX2) at (2.0,-2.4)   {$X$};
\node[node] (aZ2) at (2.75,-1.2)  {$Z$};
\node[node] (aY2) at (3.5,-2.4)   {$Y$};

\draw[edge] (aX2) -- (aZ2);
\draw[edge] (aZ2) -- (aY2);

\node[graphlabel] at (2.75,-2.95)
    {$\mathcal{G}_{2}$};

\node[node] (aXD) at (0,-4.7)      {$X$};
\node[node] (aZD) at (0.75,-3.5)   {$Z$};
\node[node] (aYD) at (1.5,-4.7)    {$Y$};

\draw[edge] (aZD) -- (aYD);

\node[graphlabel] at (0.75,-5.25)
    {$\mathcal{D}$};

\node[panellabel] at (0.75,-5.9)
    {(a)};

\end{scope}

\begin{scope}[shift={(8.0,0)}]

\node[node] (bXs) at (0,0)       {$X$};
\node[node] (bZs) at (0.75,1.2)  {$Z$};
\node[node] (bYs) at (1.5,0)     {$Y$};

\draw[edge] (bXs) -- (bZs);
\draw[edge] (bZs) -- (bYs);

\node[graphlabel] at (0.75,-0.55)
    {$\mathcal{G}^{*}$};

\node[node] (bX1) at (-2.0,-2.4)  {$X$};
\node[node] (bZ1) at (-1.25,-1.2) {$Z$};
\node[node] (bY1) at (-0.5,-2.4)  {$Y$};

\draw[edge] (bX1) -- (bZ1);

\node[graphlabel] at (-1.25,-2.95)
    {$\mathcal{G}^{1}$};

\node[node] (bX2) at (2.0,-2.4)   {$X$};
\node[node] (bZ2) at (2.75,-1.2)  {$Z$};
\node[node] (bY2) at (3.5,-2.4)   {$Y$};

\draw[edge] (bZ2) -- (bY2);

\node[graphlabel] at (2.75,-2.95)
    {$\mathcal{G}^{2}$};

\node[node] (bXD) at (0,-4.7)      {$X$};
\node[node] (bZD) at (0.75,-3.5)   {$Z$};
\node[node] (bYD) at (1.5,-4.7)    {$Y$};

\draw[edge] (bXD) -- (bZD);
\draw[edge] (bZD) -- (bYD);

\node[graphlabel] at (0.75,-5.25)
    {$\mathcal{D}$};

\node[panellabel] at (0.75,-5.9)
    {(b)};

\end{scope}

\end{tikzpicture}

\caption{Illustration of diff-separation.}
\label{fig:dif-sep}
\end{figure}

\begin{definition}[Diff-relevant path]
A path in $\mathcal{G}^1 \cup \mathcal{G}^2$ is \emph{diff-relevant} if :
\begin{itemize}
    \item it contains at least one edge in the difference DAG $\mathcal{D}$, or 
    \item  it contains a non-endpoint node that has an ancestor with at least one incoming edge in $\mathcal{D}$.
\end{itemize}
\end{definition}

\begin{definition}[Conditionally diff-relevant path]
Given a conditioning set $\mathbb{Z}$, a path $\pi$ between $X_i$ and $X_j$ in
$\mathcal{G}^1\cup\mathcal{G}^2$ is \emph{conditionally diff-relevant relative to $\mathbb{Z}$} if $\pi$ is not diff-relevant and at least one of the following conditions holds:
\begin{enumerate}
    \item $\pi$ contains a collider having a descendant in $\mathbb{Z}$ with at least one incoming edge in $\mathcal{D}$;
    \item $\pi$ contains a non-collider node that is connected, through an active path given $\mathbb{Z}$, to a node in $\mathbb{Z}\setminus \pi$ with at least one incoming edge in $\mathcal{D}$;
    \item one endpoint of $\pi$ has an incoming edge along $\pi$ and has a descendant in $\mathbb{Z}$ with at least one incoming edge in $\mathcal{D}$.
\end{enumerate}
\end{definition}

\begin{definition}[Diff-separation]
\label{def:diff_sep}
Let $\mathcal{M}^*$ be an SCM. Consider two environments $e^1$ and $e^2$, where the respective SCMs $\mathcal{M}^1$ and $\mathcal{M}^2$ are obtained by applying interventions to $\mathcal{M}^*$.
We denote $\mathcal{G}^1$ and $\mathcal{G}^2$ the corresponding causal graphs, and $\mathcal{D}$ the corresponding difference DAG. 
A set $\mathbb{Z}$ diff-separates $X_i$ from $X_j$ if:
\begin{itemize}
    \item it blocks all difference-relevant paths between $X_i$ and $X_j$;
    \item it blocks all conditionally diff-relevant paths (relative to $\mathbb{Z}$) between $X_i$ and $X_j$.
\end{itemize}
\end{definition}

In the following, we present several propositions that characterize important properties of diff-separation.
The first proposition establishes a direct link between diff-separation and the structure of the difference DAG.

\begin{restatable}{proposition}{diff}
\label{lemma:diff-set_parents}
 Let $\mathcal{M}^*$ be an SCM. Consider two environments $e^1$ and $e^2$, where the respective SCMs $\mathcal{M}^1$ and $\mathcal{M}^2$ are obtained by applying interventions to $\mathcal{M}^*$.
We denote  $\mathcal{G}^1$ and $\mathcal{G}^2$ the corresponding causal graphs, and 
$\mathcal{D}=(\mathbb{V}, \mathbb{E})$  the difference DAG of  $\mathcal{M}^1$ and $\mathcal{M}^2$.
For any two nodes $(X_i, X_j) \in \mathbb{V}^2$, there exists a subset $\mathbb{Z} \subseteq \mathbb{V} \setminus \{X_i, X_j\}$ that diff-separates $X_i$ from $X_j$ in $(\mathcal{G}^1, \mathcal{G}^2)$ if and only if $X_i \notin Pa_{\mathcal{D}}(X_j)$ and $X_j \notin Pa_{\mathcal{D}}(X_i)$.
\end{restatable}
The second proposition provide a link between diff-separation and the equality of regressions coefficients across environments.

\begin{restatable}{proposition}{diffregcoef}
\label{lemma:diff-reg_coef}
Let $\mathcal{M}^*$ be a linear SCM. Consider two environments $e^1$ and $e^2$, where the respective SCMs $\mathcal{M}^1$ and $\mathcal{M}^2$ are obtained by applying interventions to $\mathcal{M}^*$.
Assume that all exogenous noise distributions are invariant across environments.
We denote  $\mathcal{G}^1$ and $\mathcal{G}^2$ the corresponding causal graphs, and 
$\mathcal{D}=(\mathbb{V}, \mathbb{E})$  the difference DAG of  $\mathcal{M}^1$ and $\mathcal{M}^2$.
For any two nodes $(X_i, X_j) \in \mathbb{V}^2$, there exists a subset $\mathbb{Z} \subseteq \mathbb{V} \setminus \{X_i, X_j\}$ that diff-separates $X_i$ from $X_j$ in $(\mathcal{G}^1, \mathcal{G}^2)$ then 
$ {r_{X_iX_j.\mathbb{Z}}}_{\mathbb{P}^1} ={r_{X_iX_j.\mathbb{Z}}}_{\mathbb{P}^2}$.
\end{restatable}
The last proposition of this section is trivial and it emphasizes  diff-separation  is symmetric.

\begin{restatable}{proposition}{diffsymetric}
\label{lemma:diff_symetric}
There exists a set $\mathbb{Z}_1$ that diff-separates $X_i$ from $X_j$ if and only if there exists a set $\mathbb{Z}_2$ that diff-separates $X_j$ from $X_i$.
\end{restatable}

\subsection{Faithfulness assumption for difference graphs discovery}
Traditional constraint-based methods, which aim to discover a causal DAG from observational data, typically rely on the faithfulness assumption. In essence, faithfulness ensures that d-separation in the graph corresponds to statistical independence in the data. Formally, under faithfulness, $X_i$ and $X_j$ are d-separated by $\mathbb{Z}$ if and only if 
 $X_i$ and $X_j$ are conditionally independent given  $\mathbb{Z}$.
However, when constructing a difference graph based on dependence constraints, the faithfulness assumption is not required, as illustrated in Example \ref{ex:1}. 

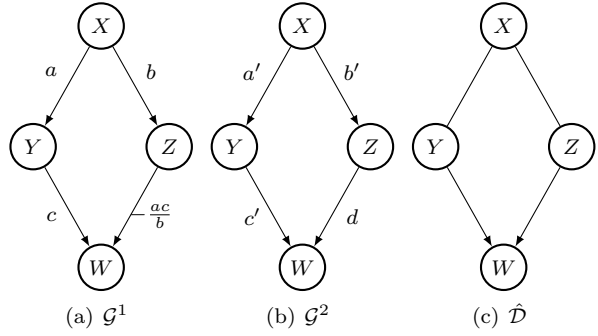
\begin{figure}[t]
\centering
 \subfigure[$\mathcal{G}^1$]{
 \label{fig:intervention_2}
  \centering
\begin{tikzpicture}[{black, circle, draw, inner sep=0, font=\footnotesize, scale=1.2}]
\tikzset{nodes={draw,rounded corners, thick},minimum height=0.6cm,minimum width=0.6cm}
		\node (X) at (0.75,2) {${X}$};
		\node (Y) at (0,0.67) {${Y}$} ;
		\node (Z) at (1.5,0.67) {${Z}$};
		\node (W) at (0.75,-0.66) {${W}$};
    	\node[draw=none] (xy) at (0.2,1.5) {$a$};
    	\node[draw=none] (xz) at (1.3,1.5) {$b$};
    	\node[draw=none] (yw) at (0.2,-0.1) {$c$};
    	\node[draw=none] (zw) at (1.3,-0.1) {$-\frac{ac}{b}$};
      \draw[->,>=latex] (X)  to (Z);
    	\draw[->,>=latex] (X) -- (Y);
         \draw[->,>=latex] (Z)  to  (W);
		\draw[->,>=latex] (Y) to (W);
    \end{tikzpicture}}
 \subfigure[$\mathcal{G}^2$]{\label{fig:diff1}
 \centering
\begin{tikzpicture}[{black, circle, draw, inner sep=0, font=\footnotesize, scale=1.2}]
\tikzset{nodes={draw,rounded corners, thick},minimum height=0.6cm,minimum width=0.6cm}
		\node (X) at (0.75,2) {${X}$};
		\node (Y) at (0,0.67) {${Y}$} ;
		\node (Z) at (1.5,0.67) {${Z}$};
		\node (W) at (0.75,-0.66) {${W}$};
     	\node[draw=none] (xy) at (0.2,1.5) {$a'$};
    	\node[draw=none] (xz) at (1.3,1.5) {$b'$};
    	\node[draw=none] (yw) at (0.2,-0.1) {$c'$};
    	\node[draw=none] (zw) at (1.3,-0.1) {$d$};
      \draw[->,>=latex] (X)  to (Z);
    	\draw[->,>=latex] (X) -- (Y);
         \draw[->,>=latex] (Z)  to  (W);
		\draw[->,>=latex] (Y) to (W);
    \end{tikzpicture}}
 \subfigure[$\hat{\mathcal{D}}$]{\label{fig:diff2}
\begin{tikzpicture}[{black, circle, draw, inner sep=0, font=\footnotesize, scale=1.2}]
\tikzset{nodes={draw,rounded corners, thick},minimum height=0.6cm,minimum width=0.6cm}
		\node (X) at (0.75,2) {${X}$};
		\node (Y) at (0,0.67) {${Y}$} ;
		\node (Z) at (1.5,0.67) {${Z}$};
		\node (W) at (0.75,-0.66) {${W}$};
      \draw[-,>=latex] (X)  to (Z);
    	\draw[-,>=latex] (X) to (Y);
         \draw[->,>=latex] (Z)  to  (W);
		\draw[->,>=latex] (Y) to (W);
    \end{tikzpicture}}
\caption{Violation of faithfulness but not diff-faithfulness. PC algorithm will not be able to infer $\mathcal{G}^1$, but LDiffPC is able to infer the partially oriented $\hat{\mathcal{D}}$ of the true difference DAG $\mathcal{D}$. }
 \label{fig:graph_faith}
\end{figure}

\begin{example}\label{ex:1}
Consider a linear SCM $\mathcal{M}^*$ and two environments, each defining an SCM, $\mathcal{M}^1$ and $\mathcal{M}^2$, associated with the DAGs $\mathcal{G}^1$ and $\mathcal{G}^2$, shown in Figure~\ref{fig:graph_faith} (a) and (b). The direct causal effects in $\mathcal{M}^1$ and $\mathcal{M}^2$ are indicated on the edges of $\mathcal{G}^1$ and $\mathcal{G}^2$. Additionally, Figure~\ref{fig:graph_faith} (c) illustrates the difference DAG $\mathcal{D}$, which captures changes between $\mathcal{M}^1$ and $\mathcal{M}^2$.

In $\mathcal{G}^1$, there exist two distinct paths between nodes $X$ and $W$ with causal effects that cancel each other out. As a result, a standard independence-based method might incorrectly infer that $X$ and $W$ are conditionally independent, leading to the erroneous removal of the edge between them.

However, if the causal coefficients are slightly modified in the second environment, such that the cancellation no longer holds exactly, the statistical dependencies between $X$ and $W$ may change. The LDiffPC method accounts for such differences and retains the edge between $X$ and $W$, unless conditioning on both intermediate variables, $Y$ and $Z$, confirms the equality of the regression coefficients between $X$ and $W$ is the same across environments.
\end{example}

By directly identifying the difference graph without first reconstructing $\mathcal{G}^1$ and $\mathcal{G}^2$, we avoid relying on the faithfulness assumption, as our approach does not depend on exact path cancellations. This makes it more robust in scenarios where small variations in causal strength could otherwise obscure true dependencies. Additionally, we do not require the minimality assumption, which states that every edge in a DAG corresponds to a dependency in the observed data distribution. By relaxing this constraint, our method remains effective even in complex or over-specified models, without enforcing that the causal graph be in its minimal form.

Although our approach does not rely on traditional constraint-based causal discovery assumptions, it does depend on an alternative assumption that we call \emph{diff-faithfulness}. This assumption is crucial for ensuring that the true difference graph can be identified. Specifically, differential faithfulness posits that, under any intervention, the dependence measure between the intervened variable and its ancestors changes in a detectable manner. This guarantees that causal modifications remain distinguishable across environments.

\begin{assumption}[Diff-faithfulness]\label{assum:interv_faithfulness}
 Let $\mathcal{M}^*$ be an SCM. Consider two environments $e^1$ and $e^2$, where the respective linear SCMs $\mathcal{M}^1$ and $\mathcal{M}^2$ are obtained by applying interventions to $\mathcal{M}^*$.
We denote  $\mathcal{G}^1$ and $\mathcal{G}^2$ the corresponding causal graphs and $\mathbb{P}^1$ and $\mathbb{P}^2$ the associate distributions. 

We assume that 
$$ {r_{X_iX_j.Z}}_{\mathbb{P}^1} ={r_{X_iX_j.Z}}_{\mathbb{P}^2} $$
if and only if 
$\mathbb{Z}$ \textit{diff}-separates $X_i$ from $X_j$.
\end{assumption}

This assumption ensures that causal interventions have detectable effects, enabling the construction of the difference DAG by accurately capturing where and how dependencies vary between environments. 

\begin{figure}
\centering
  \subfigure[$\mathcal{G}^1$]{\label{fig:intervention}
  \centering
\begin{tikzpicture}[{black, circle, draw, inner sep=0, font=\footnotesize, scale=1.2}]
\tikzset{nodes={draw,rounded corners, thick},minimum height=0.6cm,minimum width=0.6cm}
		\node (X) at (0.75,2) {${X}$};
		\node (Y) at (0,0.67) {${Y}$} ;
		\node (Z) at (1.5,0.67) {${Z}$};
		\node (W) at (0.75,-0.66) {${W}$};
    	\node[draw=none] (xy) at (0.2,1.5) {$a$};
    	\node[draw=none] (xz) at (1.3,1.5) {$b$};
    	\node[draw=none] (yw) at (0.2,-0.1) {$c$};
    	\node[draw=none] (zw) at (1.3,-0.1) {$d$};
      \draw[->,>=latex] (X)  to (Z);
    	\draw[->,>=latex] (X) -- (Y);
         \draw[->,>=latex] (Z)  to  (W);
		\draw[->,>=latex] (Y) to (W);
    \end{tikzpicture}}
  \subfigure[$\mathcal{G}^2$]{\label{fig:diff5}
  \centering
\begin{tikzpicture}[{black, circle, draw, inner sep=0, font=\footnotesize, scale=1.2}]
\tikzset{nodes={draw,rounded corners, thick},minimum height=0.6cm,minimum width=0.6cm}
		\node (X) at (0.75,2) {${X}$};
		\node (Y) at (0,0.67) {${Y}$} ;
		\node (Z) at (1.5,0.67) {${Z}$};
		\node (W) at (0.75,-0.66) {${W}$};
    	\node[draw=none] (xy) at (0.2,1.5) {$b$};
    	\node[draw=none] (xz) at (1.3,1.5) {$a$};
    	\node[draw=none] (yw) at (0.2,-0.1) {$d$};
    	\node[draw=none] (zw) at (1.3,-0.1) {$c$};
      \draw[->,>=latex] (X)  to (Z);
    	\draw[->,>=latex] (X) -- (Y);
         \draw[->,>=latex] (Z)  to  (W);
		\draw[->,>=latex] (Y) to (W);
    \end{tikzpicture}}
\subfigure[$\hat{\mathcal{D}}$]{\label{fig:diff4}\centering
\begin{tikzpicture}[{black, circle, draw, inner sep=0, font=\footnotesize, scale=1.2}]
\tikzset{nodes={draw,rounded corners, thick},minimum height=0.6cm,minimum width=0.6cm}
		\node (X) at (0.75,2) {${X}$};
		\node (Y) at (0,0.67) {${Y}$} ;
		\node (Z) at (1.5,0.67) {${Z}$};
		\node (W) at (0.75,-0.66) {${W}$};
      \draw[->,>=latex] (X)  to (Z);
    	\draw[->,>=latex] (X) -- (Y);
         \draw[->,>=latex] (W)  to  (Z);
		\draw[->,>=latex] (W) to (Y);
    \end{tikzpicture}}
\caption{Violation of diff-faithfulness through swapping of coefficients. 
}
 \label{fig:graph_int_faith}
\end{figure}
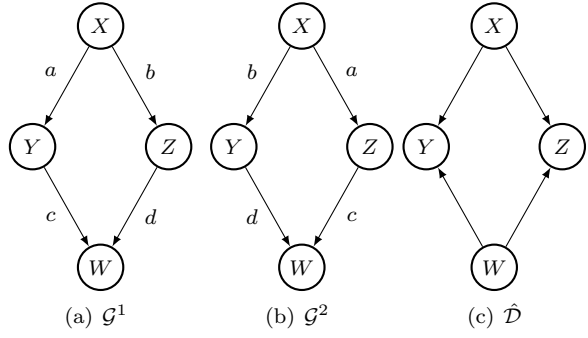

The diff-faithfulness assumption also implies that identical canceling paths (i.e., violations of faithfulness) cannot persist across both causal graphs, and it restricts swapping causal coefficients between graphs. This concept is illustrated in Figure \ref{fig:graph_int_faith}.

\section{Constraint-based approach for discovering the difference DAG}\label{sec:method}

To construct the difference graph, we propose LDiffPC, a constraint-based algorithm inspired by the PC algorithm~\citep{Spirtes_2000}. Whereas the classical PC algorithm relies on d-separation to recover a causal graph, LDiffPC relies on diff-separation and equality tests of regression coefficients to infer the structure of the difference graph.

Before presenting the main algorithm, we first show that, under the diff-faithfulness assumption, the skeleton of the true difference graph is identifiable from equality tests of regression coefficients. This is established in the following lemma.

\begin{restatable}{lemma}{skeleton}
\label{lemma:interventional_faithfulness_implications_skeleton}
 Let $\mathcal{M}^*$ be an SCM. Consider two environments $e^1$ and $e^2$, where the respective SCMs $\mathcal{M}^1$ and $\mathcal{M}^2$ are obtained by applying interventions to $\mathcal{M}^*$.
We denote $\mathbb{P}^1$ and $\mathbb{P}^2$ the associate distributions, and  $\mathcal{D} = (\mathbb{V}, \mathbb{E})$  the true difference graph between  $\mathcal{M}^1$ and $\mathcal{M}^2$. 
If Assumption~\ref{assum:interv_faithfulness} is satisfied then for all pair of vertices $(X_i,X_j) \in \mathbb{V}^{2}$, $X_i$ and $X_j$ are adjacent in $\mathcal{D}$ if and only if  for any subset $\mathbb{Z} \subseteq \mathbb{V}\backslash\{X_i,X_j\}$, $r_{X_iX_j.\mathbb{Z}\mathbb{P}^1} \neq r_{X_iX_j.\mathbb{Z}\mathbb{P}^2}$. 
\end{restatable}

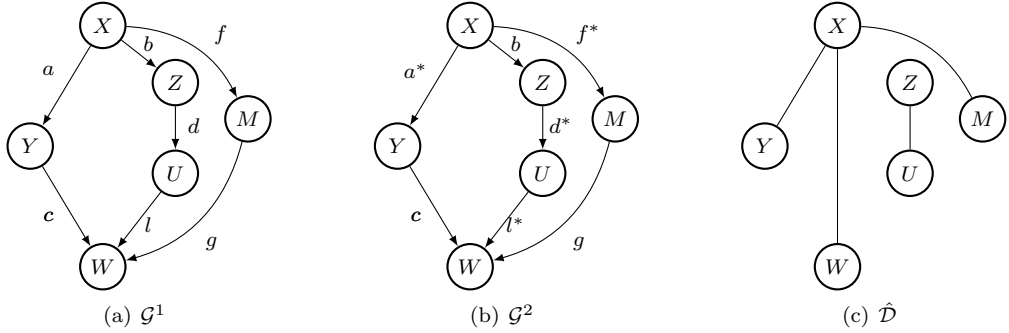
\begin{figure*}[t]
\centering
 \subfigure[$\mathcal{G}^1$]{
 \label{fig:intervention_2}
  \centering
\begin{tikzpicture}[{black, circle, draw, inner sep=0, font=\footnotesize, scale=1.2}]
		\tikzset{nodes={draw,rounded corners, thick},minimum height=0.6cm,minimum width=0.6cm}
		\node (X) at (0.8,2) {${X}$};
		\node (Y) at (0,0.67) {${Y}$} ;
		\node (Z) at (1.6,1.37) {${Z}$};
		\node (W) at (0.8,-0.66) {${W}$};
        \node (U) at (1.6,0.37) {${U}$};
         \node (M) at (2.4,0.97) {${M}$};
    	\node[draw=none] (xy) at (0.2,1.5) {$a$};
    	\node[draw=none] (xz) at (1.3,1.8) {$b$};
    	\node[draw=none] (yw) at (0.2,-0.1) {$c$};
     \node[draw=none] (yw) at (0.2,-0.1) {$c$};
     \node[draw=none] (zu) at (1.8,0.9) {$d$};
       \node[draw=none] (xm) at (2.1,1.9) {$f$};
        \node[draw=none] (mw) at (2,-0.4) {$g$};
          \node[draw=none] (uw) at (1.3,-0.2) {$l$};
      \draw[->,>=latex] (X)  to (Z);
         \draw[->,>=latex, bend left] (X)  to (M);
               \draw[->,>=latex, bend left] (M)  to (W);
    	\draw[->,>=latex] (X) -- (Y);
         \draw[->,>=latex] (Z)  to  (U);
           \draw[->,>=latex] (U)  to  (W);
		\draw[->,>=latex] (Y) to (W);
    \end{tikzpicture}}
    \hspace{1cm}
 \subfigure[$\mathcal{G}^2$]{\label{fig:diff1}
 \centering
\begin{tikzpicture}[{black, circle, draw, inner sep=0, font=\footnotesize, scale=1.2}]
\tikzset{nodes={draw,rounded corners, thick},minimum height=0.6cm,minimum width=0.6cm}
		\node (X) at (0.8,2) {${X}$};
		\node (Y) at (0,0.67) {${Y}$} ;
		\node (Z) at (1.6,1.37) {${Z}$};
		\node (W) at (0.8,-0.66) {${W}$};
        \node (U) at (1.6,0.37) {${U}$};
         \node (M) at (2.4,0.97) {${M}$};
    	\node[draw=none] (xy) at (0.2,1.5) {$a^{*}$};
    	\node[draw=none] (xz) at (1.3,1.8) {$b$};
    	\node[draw=none] (yw) at (0.2,-0.1) {$c$};
     \node[draw=none] (yw) at (0.2,-0.1) {$c$};
     \node[draw=none] (zu) at (1.8,0.9) {$d^{*}$};
       \node[draw=none] (xm) at (2.1,1.9) {$f^{*}$};
        \node[draw=none] (mw) at (2,-0.4) {$g$};
           \node[draw=none] (uw) at (1.3,-0.2) {$l^{*}$};
      \draw[->,>=latex] (X)  to (Z);
         \draw[->,>=latex, bend left] (X)  to (M);
               \draw[->,>=latex, bend left] (M)  to (W);
    	\draw[->,>=latex] (X) -- (Y);
         \draw[->,>=latex] (Z)  to  (U);
           \draw[->,>=latex] (U)  to  (W);
		\draw[->,>=latex] (Y) to (W);
    \end{tikzpicture}}
        \hspace{1cm}
\subfigure[$\hat{\mathcal{D}}$]{\label{fig:diff2}
\begin{tikzpicture}[{black, circle, draw, inner sep=0, font=\footnotesize, scale=1.2}]
\tikzset{nodes={draw,rounded corners, thick},minimum height=0.6cm,minimum width=0.6cm}
		\node (X) at (0.8,2) {${X}$};
		\node (Y) at (0,0.67) {${Y}$} ;
		\node (Z) at (1.6,1.37) {${Z}$};
		\node (W) at (0.8,-0.66) {${W}$};
        \node (U) at (1.6,0.37) {${U}$};
         \node (M) at (2.4,0.97) {${M}$};
      \draw[-,>=latex] (X)  to (W);
         \draw[-,>=latex, bend left] (X)  to (M);
    	\draw[-,>=latex] (X) -- (Y);
           \draw[-,>=latex] (Z)  to  (U);
    \end{tikzpicture}}
\caption{Example of the error of LDiffPC algorithm with PC adjacency search. 
}
 \label{fig:PC_fail}
\end{figure*}

Beyond recovering the skeleton, it is necessary to show that all v-structures can be correctly identified given access to all equality tests. This is established in the following lemma.

\begin{restatable}{lemma}{vstructure}
\label{lemma:interventional_faithfulness_implications_v_structure}
 Let $\mathcal{M}^*$ be an SCM. Consider two environments $e^1$ and $e^2$, where the respective SCMs $\mathcal{M}^1$ and $\mathcal{M}^2$ are obtained by applying interventions to $\mathcal{M}^*$.
We denote $\mathbb{P}^1$ and $\mathbb{P}^2$ the associate distributions, and  $\mathcal{D} = (\mathbb{V}, \mathbb{E})$  the true difference graph between  $\mathcal{M}^1$ and $\mathcal{M}^2$. 
If Assumption~\ref{assum:interv_faithfulness} is satisfied then for all triplets of vertices $(X_i,X_k,X_j) \in \mathbb{V}^3$ such that $X_i$ and $X_k$ are adjacent, $X_k$ and $X_j$ are adjacent, and $X_i$ and $X_j$ are not adjacent, $X_i\rightarrow X_k \leftarrow X_j$ if and only if for  any subset $\mathbb{Z} \subseteq \mathbb{V}\backslash\{X_i,X_j\}$ such that $X_k\in \mathbb{Z}$, $r_{X_iX_j.\mathbb{Z}\mathbb{P}^1} \neq r_{X_iX_j.\mathbb{Z}\mathbb{P}^2}$.
\end{restatable}

We refer to the graph obtained after these two steps as the pattern of the difference graph—denoted $pattern(\mathcal{D})$—which includes the correct skeleton with all V-structures oriented.
In the following, we show that the same rules using the PC algorithms to some edges orient edges that are not part of a V-structure are also sound in our setting.

\begin{definition}
Consider a difference graph $\mathcal{D} = (\mathbb{V}, \mathbb{E})$ and its pattern $pattern(\mathcal{D})$. Meek-rules are the following: for $(X_i, X_j, X_k) \in \mathbb{V}^3$,
\begin{enumerate}
    \item[Rule 1:] If $X_i\rightarrow X_k-X_j$ and $X_i$ 
    is in $pattern(\mathcal{D})$ and $X_j$ are not adjacent,  then $X_i\rightarrow X_k \rightarrow X_j$ is in $\mathcal{D}$; 
    \item[Rule 2:]  If there is a directed path from $X_i$ to  $X_j$ and an edge  $X_i - X_j$ in $pattern(\mathcal{D})$  then  $X_i \rightarrow X_j$  in $\mathcal{D}$; 
    \item[Rule 3:] If there are two paths  $X_i -X_l \rightarrow X_j$ and  $X_i -X_m \rightarrow X_j$ and an edge $X_i - X_j$ in  $pattern(\mathcal{D})$ , then $X_j \rightarrow  X_j$.
\end{enumerate}
\end{definition}

\begin{lemma}
    \label{lemma:meek}
Consider a difference graph $\mathcal{D} = (\mathbb{V}, \mathbb{E})$ and its pattern $pattern(\mathcal{D})$. 
    Then, successive applications of Meek-Rules in $pattern(\mathcal{D})$ yield correct orientations.
\end{lemma}
\begin{proof}
As $\mathcal{D}$ is a DAG and by Lemma~\ref{lemma:interventional_faithfulness_implications_v_structure} all the V-structures are detected in the $pattern(\mathcal{D})$, the proof is identical to the proof in \citep[Theorem 2]{Meek_1995}.
\end{proof}

Notice that the three lemmas above naturally suggest a naive algorithm that is trivially sound. We refer to this procedure as the brute-force version of LDiffPC. It performs all possible equality tests over all conditioning sets. This approach follows the same philosophy as the SGS algorithm, which can be viewed as the exhaustive counterpart of the PC algorithm. The pseudocode of the brute-force algorithm is provided in Algorithm~\ref{algo:BFIDI} in Appendix~\ref{app:BF}. 

In classical causal discovery, the PC algorithm~\citep{Spirtes_2000, Colombo_2014} was introduced to improve the efficiency of the first two steps of the SGS algorithm. It achieves this by: (i) considering conditioning sets of increasing size, from $0$ up to $p-2$; (ii) restricting conditioning sets to adjacent nodes, thereby reducing the number of tests; and (iii) storing separation sets, which are later used to orient v-structures efficiently. In the final step, similarly to SGS, the algorithm further orients edges using the rules introduced by \cite{Meek_1995}. 

In the following, we show how the brute-force version can be optimized to obtain LDiffPC, in analogy with how PC improves upon SGS. Our algorithm begins by constructing the skeleton: we start from a fully connected undirected graph and iteratively remove edges based on equality tests of regression coefficients. However, the PC algorithm cannot be directly adapted to the problem of difference DAG inference. The following example illustrates this limitation.

\begin{example}
Consider two SCMs, $\mathcal{M}^{1}$ and $\mathcal{M}^{2}$, defined over the variables $X, Y, Z, W, U,$ and $M$. Their corresponding causal graphs, $\mathcal{G}^1$ and $\mathcal{G}^2$, are depicted in Figure~\ref{fig:PC_fail}.
We assume interventions on the nodes $Y, U,$ and $M$, meaning that the direct causal relationships have changed for the following pairs: $(X, Y)$,  $(Z, U)$, $(X, M)$.

Our algorithm starts from a fully connected, undirected graph. Figure~\ref{fig:PC_fail}(c) illustrates the result of applying a naive adaptation of the PC algorithm, which restricts conditioning sets to adjacent nodes. However, this approach fails in this setting.
Specifically, the edge between $X$ and $W$ cannot be removed when conditioning only on their adjacent nodes, because $Z$ and $U$—which are essential for separating $X$ and $W$—are no longer adjacent to them. In this case, correctly removing the edge between $X$ and $W$ would require conditioning on $Y, M,$ and either $Z$ or $U$.
\end{example}

This example highlights a fundamental limitation: the adjacency-based conditioning strategy of the PC algorithm is insufficient for difference DAG inference, as it may fail to recover the correct skeleton. Nevertheless, some key ideas from the PC algorithm can still be leveraged. In particular, we can refine the search space by: (i) using a sequential conditioning strategy, where the size of the conditioning sets increases progressively from $0$ to $p-2$; and (ii) storing separation sets, which are later used to orient v-structures and other configurations. 

Algorithm~\ref{algo:IDI} presents the Linear Difference PC (LDiffPC) algorithm, which summarizes the overall procedure. The resulting graph is only partially oriented, as equality tests alone are not sufficient to fully orient all edges. The output is therefore a partially directed difference graph. The following theorem shows that LDiffPC correctly recovers the true partially directed difference graph, including its skeleton, v-structures, and all orientations implied by Meek’s rules.

\begin{algorithm}[t]
\caption{LDiffPC algorithm}
\label{algo:IDI}
\SetKwInOut{Input}{Input}
 \Input{Two datasets $(\mathbf{x}_1^1,\ldots, \mathbf{x}_{n_1}^1), (\mathbf{x}_1^2,\ldots, \mathbf{x}_{n_2}^2)$, a measure of dependence $I$, a dependence test}
 	\KwResult{The difference PDAG $\mathcal{\hat D}$}
 Form the complete undirected graph $\hat{\mathcal{D}}$ on the vertex set $\mathbb{V}$.

 $n \leftarrow 0$
            
    \Repeat{
	all subsets $\mathbb{Z}$  of cardinality $n$  have been tested  }{
     \For{each  pair of adjacent vertices  $X_i - X_j$ in $\hat{\mathcal{D}}$}{ 
    	 \If{for any subset $\mathbb{Z} \subseteq \mathbb{V}\backslash\{X_i,X_j\}$ of cardinality $n$, $r_{X_i, X_j.\mathbb{Z}\mathbb{P}^1} = r_{X_i, X_j.\mathbb{Z}\mathbb{P}^2}$} 
            {  
            
            delete edge $X_i - X_j$ from $\hat{\mathcal{D}}$
                        
            record $\mathbb{Z}$ in $Sepset(X_i,X_j)$ and $Sepset(X_j,X_i)$     }
         } 
                 $n \leftarrow n + 1$
        }

 \For{ each unshielded triplet  $X_i-X_k-X_j$ in $\hat{\mathcal{D}}$}{ 
     	 \If{$X_k\not\in Sepset(X_i,X_j)$}{
         Orient the triplet as $X_i\rightarrow X_k\leftarrow X_j$ in $\hat{\mathcal{D}}$
         } 
}

\Repeat {no edge can be oriented 
}        {Apply Meek-Rules}
		      
\end{algorithm}

\begin{restatable}{theorem}{LDiffPCtheorem}\label{thm1}
  Let $\mathcal{M}^*$ be an SCM. Consider two environments $e^1$ and $e^2$, where the respective SCMs $\mathcal{M}^1$ and $\mathcal{M}^2$ are obtained by applying interventions to $\mathcal{M}^*$.
Let  $\mathcal{D} = (\mathbb{V}, \mathbb{E})$ be the true difference graph between  $\mathcal{M}^1$ and $\mathcal{M}^2$, and its pattern $(\mathcal{D})$. 
Suppose $\mathcal{\hat D}$ is the output of  the algorithm LDiffPC.
Under Assumption~\ref{assum:interv_faithfulness}  and assuming access to perfect conditional equality information for all pairs of variables
in $\mathbb{V}$, 
$\mathcal{\hat D}$ and $\mathcal{D}$ have the same skeleton and all orientations present in $\mathcal{\hat D}$ are also present in $\mathcal{D}$.
\end{restatable}

\section{Testing conditional equalities} \label{sec:perm}

 In practice, we rely on the regression coefficient
 $r_{XY.\mathbb{Z}\mathbb{P}^e}$ as our measure of conditional equality. 
 To implement this approach, we assume we have two samples, $(\mathbf{x}_1^1,\ldots, \mathbf{x}_{n_1}^1)$ and $(\mathbf{x}_1^2,\ldots, \mathbf{x}_{n_2}^2)$, drawn from two SCMs $\mathcal{M}^1$ and $\mathcal{M}^2$, respectively, where $\mathbf{x}^\ell_i \in \mathbb{R}^p$ for $\ell \in \{1,2\}$. We aim to test whether the regression coefficients between pairs of variables, conditional on a set of variables, are equivalent across the two populations.

 We first calculate the observed difference in conditional regression coefficients: for $i,j \in \{1,\ldots,p\}$ and $\mathbb{Z} \subseteq \{X_1,\ldots,X_p\}\setminus\{X_i,X_j\}$,
 $$
 \hat{\Delta}_{\mathrm{obs}}(X_i,X_j \mid \mathbb{Z})
 =
 \hat r_{X_i X_j \cdot \mathbb{Z}\mathbb{P}^1}
 -
 \hat r_{X_i X_j \cdot \mathbb{Z}\mathbb{P}^2}.
 $$

To assess the significance of this observed difference, we employ a permutation test.
We start by combining the samples from both datasets, creating a single, pooled dataset. Then, we randomly permute the entries in this combined dataset and split it back into two groups of sizes $n_1$ and $n_2$, maintaining the original sample sizes for each population.
  For each permutation $\pi$, we calculate the difference in regression coefficient for the resampled datasets, which gives the value of $\hat{\Delta}^\pi_{\text{perm}}$. 
  We repeat this procedure over $N_r$ different permutations, yielding a distribution of permuted regression coefficient differences. 
  Finally, we calculate the p-value by determining the proportion of these permuted differences that are as extreme or more extreme than the observed difference, as follows:
$$\widehat{\text{p-value}} = \frac{\sum_{\pi = 1}^{N_r} {1}_{\hat{\Delta}^\pi_{\text{perm}} \geq \hat{\Delta}_{\mathrm{obs}}}}{N_r}.$$

This p-value allows us to evaluate whether the observed regression coefficient difference is statistically significant, supporting robust inference of causal differences across the two populations.
Another approach to testing the equality of the regression coefficients is to use hypothesis testing based on the F-test, as introduced in \cite{wang2018direct}.

\section{Conclusion}\label{sec:conc}

In this study, we introduced the LDiffPC method for causal discovery, designed to infer difference DAGs for linear SCMs. We proved that LDiffPC is both sound and complete. 
These results highlight its robustness in capturing causal structures without relying on stringent assumptions, providing a powerful and generalizable approach for researchers and practitioners in causal inference. 
Thus, beyond its theoretical guarantees, LDiffPC paves the way for more flexible causal discovery methods. Many scientific and medical analyses rely on difference DAGs, and LDiffPC offers a principled way to infer these structures while maintaining statistical rigor. 

This work also opens new research directions. For example, by relaxing topological ordering constraints, we may gain deeper insights into causal relationships, although this remains a challenging problem.
Another promising extension is the incorporation of latent variables, following the spirit of FCI, but adapted to the inference of difference graphs. These perspectives reinforce the potential impact of LDiffPC, not only as a theoretical contribution but also as a practical tool for causal discovery in complex environments.

\bibliography{references}

\appendix
\newpage

\onecolumn

\begin{center}
{\LARGE\bfseries Constraint-based difference graph discovery
in a linear setting\\[0.5em]
(Supplementary Material)\par}
\end{center}

\vspace{2em}

\section{Proof of Lemma \ref{lemma:diff-set_parents}} \label{proof:lemma}
 \diff*

\begin{proof}
We prove both directions.

$(\Rightarrow)$ Suppose that there exists a set
$\mathbb{Z}\subseteq \mathbb{V}\setminus\{X_i,X_j\}$ that diff-separates
$X_i$ from $X_j$.
Assume, for contradiction, that $X_i\in Pa_{\mathcal D}(X_j)$. Then
$X_i\to X_j$ is an edge in $\mathcal D$, and hence the direct path
$X_i\to X_j$ is difference-relevant. Since neither $X_i$ nor $X_j$ belongs
to $\mathbb{Z}$, this direct path cannot be blocked by $\mathbb{Z}$,
contradicting the fact that $\mathbb{Z}$ diff-separates $X_i$ from $X_j$.
Thus $X_i\notin Pa_{\mathcal D}(X_j)$. The same argument applied to the edge
$X_j\to X_i$ gives $X_j\notin Pa_{\mathcal D}(X_i)$.

$(\Leftarrow)$ Suppose now that
$X_i\notin Pa_{\mathcal D}(X_j)$ and
$X_j\notin Pa_{\mathcal D}(X_i)$. Then there is no direct edge between
$X_i$ and $X_j$ in $\mathcal D$.
Let $\mathcal F_{ij}$ be the family of sets
$\mathbb{Z}\subseteq \mathbb{V}\setminus\{X_i,X_j\}$ that block all
diff-relevant and conditionally diff-relevant paths between $X_i$ and $X_j$
relative to $\mathbb{Z}$. Since there is no direct diff-relevant path between
$X_i$ and $X_j$, every relevant path that must be blocked is non-direct.
Thus, by the finiteness of $\mathcal{G}^1\cup\mathcal{G}^2$, there exists at
least one set in $\mathcal F_{ij}$.
Let $\mathbb{Z}^\star$ be a minimal element of $\mathcal F_{ij}$. By
construction, $\mathbb{Z}^\star$ blocks all diff-relevant paths between
$X_i$ and $X_j$, and all conditionally diff-relevant paths relative to
$\mathbb{Z}^\star$. Therefore, by Definition~\ref{def:diff_sep},
$\mathbb{Z}^\star$ diff-separates $X_i$ from $X_j$.

Therefore, there exists a set
$\mathbb{Z}\subseteq\mathbb{V}\setminus\{X_i,X_j\}$ that diff-separates
$X_i$ from $X_j$ if and only if
$X_i\notin Pa_{\mathcal D}(X_j)$ and
$X_j\notin Pa_{\mathcal D}(X_i)$.
\end{proof}

\begin{lemma}[Environment-dependent path contributions]
\label{lemma:env_dependent_paths}
Let $\mathcal{M}^1$ and $\mathcal{M}^2$ be two linear SCMs with invariant exogenous noise distributions, and let $\mathcal{D}$ be their difference DAG. Fix two nodes $X_i,X_j$ and a conditioning set $\mathbb{Z}\subseteq \mathbb{V}\setminus\{X_i,X_j\}$. If either the residual covariance
$$
\operatorname{Cov}_{\mathbb{P}^\ell}
\left(R_{X_i}^{\mathbb{Z}},R_{X_j}^{\mathbb{Z}}\right)
$$
or the residual variance
$$
\operatorname{Var}_{\mathbb{P}^\ell}
\left(R_{X_i}^{\mathbb{Z}}\right)
$$
contains a contribution that differs between $\mathbb{P}^1$ and $\mathbb{P}^2$, then this contribution is induced by a path between $X_i$ and $X_j$ that is either diff-relevant or conditionally diff-relevant relative to $\mathbb{Z}$.
\end{lemma}

\begin{proof}
Since the SCMs are linear, each variable can be written as a linear combination of the exogenous noise variables. Therefore, covariances, residual covariances, residual variances, and regression coefficients can be expressed as sums of path contributions. Each contribution is determined by products of structural coefficients along the relevant paths and by exogenous noise variances.

Because the exogenous noise distributions are invariant across environments, a contribution can differ between $\mathbb{P}^1$ and $\mathbb{P}^2$ only if the contribution depends on a structural coefficient that differs across environments. Such coefficients are precisely those represented by incoming edges in the difference DAG $\mathcal{D}$.

We consider how such an environment-dependent coefficient can enter the residual covariance or residual variance.

First, suppose a path between $X_i$ and $X_j$ contains an edge in $\mathcal{D}$. Then the product of coefficients along this path differs across environments. Hence the path can induce a different contribution to the residual covariance, unless it is blocked by $\mathbb{Z}$. By definition, this path is diff-relevant.

Second, suppose a path between $X_i$ and $X_j$ does not itself contain an edge in $\mathcal{D}$, but contains a non-endpoint node whose distribution changes because one of its ancestors has an incoming edge in $\mathcal{D}$. Then the variation transmitted through this internal node can differ across environments, even though the path coefficients on the path itself are unchanged. Such a change may alter the residual covariance between $X_i$ and $X_j$, or the residual variance of $X_i$ relevant to the coefficient. By definition, this path is again diff-relevant.

It remains to consider paths that are not diff-relevant. For such paths, neither the path coefficients nor the distributions of internal nodes change directly. Therefore, any environment-dependent contribution can only arise through the conditioning operation, i.e. through the way $\mathbb{Z}$ residualizes $X_i$ and $X_j$.

There are three possible graphical ways in which this can occur.

\begin{enumerate}
    \item If the path contains a collider having a descendant in $\mathbb{Z}$ with at least one incoming edge in $\mathcal{D}$, then conditioning on $\mathbb{Z}$ can open the collider path with a strength that depends on the changed mechanism. Hence the path can contribute differently across environments.

    \item If the path contains a non-collider node that is connected, through an active path given $\mathbb{Z}$, to a node in $\mathbb{Z}\setminus \pi$ with at least one incoming edge in $\mathcal{D}$, then conditioning on that node in $\mathbb{Z}$ changes how the non-collider contribution is residualized. Since the conditioning variable has an environment-dependent mechanism, the residualized contribution may differ across environments.

    \item If one endpoint of the path has an incoming edge along the path and has a descendant in $\mathbb{Z}$ with at least one incoming edge in $\mathcal{D}$, then conditioning on that descendant can change the residualized variation of the endpoint. This can affect the residual covariance or residual variance appearing in the regression coefficient.
\end{enumerate}

These three cases are exactly the cases in which a path that is not diff-relevant is conditionally diff-relevant relative to $\mathbb{Z}$. Hence any environment-dependent contribution to either the residual covariance or the residual variance must be induced by a path that is diff-relevant or conditionally diff-relevant relative to $\mathbb{Z}$.
\end{proof}

\diffregcoef*

\begin{proof}
Since the SCMs are linear, the regression coefficient of $X_i$ in the regression of $X_j$ on $X_i$ and $\mathbb{Z}$ can be written as
$$
{r_{X_iX_j\cdot \mathbb{Z}}}_{\mathbb{P}^\ell}
=
\frac{
\operatorname{Cov}_{\mathbb{P}^\ell}
\left(R_{X_i}^{\mathbb{Z}},R_{X_j}^{\mathbb{Z}}\right)
}{
\operatorname{Var}_{\mathbb{P}^\ell}
\left(R_{X_i}^{\mathbb{Z}}\right)
},
\qquad \ell \in \{1,2\},
$$
where $R_{X_i}^{\mathbb{Z}}$ and $R_{X_j}^{\mathbb{Z}}$ denote the residuals obtained after linearly projecting $X_i$ and $X_j$ on $\mathbb{Z}$.

By Lemma~\ref{lemma:env_dependent_paths}, any environment-dependent contribution to either the residual covariance
$$
\operatorname{Cov}_{\mathbb{P}^\ell}
\left(R_{X_i}^{\mathbb{Z}},R_{X_j}^{\mathbb{Z}}\right)
$$
or the residual variance term relevant to the regression coefficient
$$
\operatorname{Var}_{\mathbb{P}^\ell}
\left(R_{X_i}^{\mathbb{Z}}\right)
$$
must arise from a path between $X_i$ and $X_j$ that is either diff-relevant or conditionally diff-relevant relative to $\mathbb{Z}$.

Since $\mathbb{Z}$ diff-separates $X_i$ from $X_j$, all such paths are blocked. Hence no environment-dependent contribution remains in either the numerator or the denominator of the regression coefficient. Therefore,
$$
\operatorname{Cov}_{\mathbb{P}^1}
\left(R_{X_i}^{\mathbb{Z}},R_{X_j}^{\mathbb{Z}}\right)
=
\operatorname{Cov}_{\mathbb{P}^2}
\left(R_{X_i}^{\mathbb{Z}},R_{X_j}^{\mathbb{Z}}\right)
$$
and
$$
\operatorname{Var}_{\mathbb{P}^1}
\left(R_{X_i}^{\mathbb{Z}}\right)
=
\operatorname{Var}_{\mathbb{P}^2}
\left(R_{X_i}^{\mathbb{Z}}\right).
$$

Combining these two equalities gives
$$
{r_{X_iX_j\cdot \mathbb{Z}}}_{\mathbb{P}^1}
=
{r_{X_iX_j\cdot \mathbb{Z}}}_{\mathbb{P}^2}.
$$
\end{proof}

\diffsymetric*

\begin{proof}
We prove one direction; the converse follows by exchanging the roles of $X_i$ and $X_j$.

Assume that there exists a set $\mathbb{Z}_1$ that diff-separates $X_i$ from $X_j$. By Definition~\ref{def:diff_sep}, $\mathbb{Z}_1$ blocks all paths between $X_i$ and $X_j$ that are either diff-relevant or conditionally diff-relevant relative to $\mathbb{Z}_1$.

Since paths between $X_i$ and $X_j$ are the same as paths between $X_j$ and $X_i$, and since both diff-relevance and conditional diff-relevance relative to a fixed conditioning set are symmetric with respect to the endpoints of the path, the same set $\mathbb{Z}_1$ also blocks all diff-relevant and conditionally diff-relevant paths between $X_j$ and $X_i$.

Hence $\mathbb{Z}_1$ diff-separates $X_j$ from $X_i$. Therefore,
\[
\exists \mathbb{Z}_1 : \mathbb{Z}_1 \text{ diff-separates } X_i \text{ from } X_j
\Rightarrow
\exists \mathbb{Z}_2 : \mathbb{Z}_2 \text{ diff-separates } X_j \text{ from } X_i.
\]
The reverse implication follows by the same argument with $X_i$ and $X_j$ exchanged.
\end{proof}

\section{Brute force algorithm}
\label{app:BF}
\begin{algorithm}[h]
\caption{Brute force LDiffPC algorithm}
\label{algo:BFIDI}
\SetKwInOut{Input}{Input}
 \Input{Two datasets $(\mathbf{x}_1^1,\ldots, \mathbf{x}_n^1), (\mathbf{x}_1^2,\ldots, \mathbf{x}_n^2)$, a measure of dependence $I$}
 	\KwResult{The maximally oriented difference graph $\mathcal{\hat D}$}
 Form the complete undirected graph $\hat{\mathcal{D}}$ on the vertex set $\mathbb{V}$.

\For{each  pair of adjacent vertices  $X_i - X_j$ in $\hat{\mathcal{D}}$}{ 
    	 \If{for any subset $\mathbb{Z} \subseteq \mathbb{V}\backslash\{X_i,X_j\}$ $r_{X_i, X_j.\mathbb{Z}\mathbb{P}^1} = r_{X_i, X_j.\mathbb{Z}\mathbb{P}^2}$} 
            {  delete edge $X_i - X_j$ from $\hat{\mathcal{D}}$
     } 
}

\For{each undirected unshielded triplet  $X_i - X_k - X_j$ in $\hat{\mathcal{D}}$}{ 
    	 \If{for all subset $\mathbb{Z} \subseteq \mathbb{V}\backslash\{X_i,X_j\}$ such that $X_k\in \mathbb{Z}$, $r_{X_i, X_j.\mathbb{Z}\mathbb{P}^1} = r_{X_i, X_j.\mathbb{Z}\mathbb{P}^2}$} 
            {  orient the triplet $X_i \rightarrow X_k \leftarrow X_j$ in $\hat{\mathcal{D}}$
     } 
}

\Repeat {no edge can be oriented 
}        {Apply Meek-Rules}
\end{algorithm}

\section{Proof of Proposition \ref{lemma:interventional_faithfulness_implications_skeleton}}\label{app:proofProp1}
\skeleton*

\begin{proof}
By Lemma~\ref{lemma:diff-set_parents}, $X_i$ and $X_j$ are non-adjacent in $\mathcal{D}$ if and only if there exists a set $\mathbb{Z}\subseteq \mathbb{V}\setminus\{X_i,X_j\}$ that diff-separates $X_i$ from $X_j$.
By Assumption~\ref{assum:interv_faithfulness}, this holds if and only if there exists a set $\mathbb{Z}\subseteq \mathbb{V}\setminus\{X_i,X_j\}$ such that
$$
{r_{X_iX_j\cdot\mathbb{Z}}}_{\mathbb{P}^1}
=
{r_{X_iX_j\cdot\mathbb{Z}}}_{\mathbb{P}^2}.
$$
Taking the negation gives that $X_i$ and $X_j$ are adjacent in $\mathcal{D}$ if and only if, for every $\mathbb{Z}\subseteq \mathbb{V}\setminus\{X_i,X_j\}$,
$$
{r_{X_iX_j\cdot\mathbb{Z}}}_{\mathbb{P}^1}
\neq
{r_{X_iX_j\cdot\mathbb{Z}}}_{\mathbb{P}^2}.
$$
\end{proof}

\vstructure*

\begin{proof}
    By Assumption~\ref{assum:interv_faithfulness},
$
{r_{X_iX_j\cdot\mathbb{Z}}}_{\mathbb{P}^1}
=
{r_{X_iX_j\cdot\mathbb{Z}}}_{\mathbb{P}^2}
$
if and only if $\mathbb{Z}$ diff-separates $X_i$ from $X_j$.

Since $X_i$ and $X_j$ are not adjacent, Lemma~\ref{lemma:diff-set_parents} implies that there exists at least one set that diff-separates $X_i$ from $X_j$.

For an unshielded triple $X_i-X_k-X_j$, the middle node $X_k$ belongs to a separating set for $X_i$ and $X_j$ if and only if $X_k$ is a non-collider on the path. Therefore, $X_i\to X_k\leftarrow X_j$ holds if and only if no diff-separating set for $X_i$ and $X_j$ contains $X_k$.

Using Assumption~\ref{assum:interv_faithfulness}, this is equivalent to saying that for every $\mathbb{Z}\subseteq \mathbb{V}\setminus\{X_i,X_j\}$ with $X_k\in\mathbb{Z}$,
$$
{r_{X_iX_j\cdot\mathbb{Z}}}_{\mathbb{P}^1}
\neq
{r_{X_iX_j\cdot\mathbb{Z}}}_{\mathbb{P}^2}.
$$
This proves the claim.
\end{proof}

\newpage
\section{Proof of Theorem \ref{thm1}}\label{app:proofThm}
\LDiffPCtheorem*
\begin{proof}
We prove the result in three steps.

First, consider the skeleton recovered by LDiffPC. The algorithm removes an edge $X_i-X_j$ whenever it finds a set
$\mathbb{Z}\subseteq \mathbb{V}\setminus\{X_i,X_j\}$ such that
$$
{r_{X_iX_j\cdot \mathbb{Z}}}_{\mathbb{P}^1}
=
{r_{X_iX_j\cdot \mathbb{Z}}}_{\mathbb{P}^2}.
$$
Since the algorithm has access to perfect conditional equality information and tests all conditioning sets by increasing cardinality, Lemma~\ref{lemma:interventional_faithfulness_implications_skeleton} implies that an edge is removed if and only if $X_i$ and $X_j$ are not adjacent in the true difference graph $\mathcal{D}$. Hence, after the skeleton phase, $\hat{\mathcal D}$ and $\mathcal D$ have the same skeleton.

Second, consider the orientation of unshielded triples. Let $X_i-X_k-X_j$ be an unshielded triple in the recovered skeleton. Since the skeleton is correct, this is also an unshielded triple in $\mathcal D$. By Lemma~\ref{lemma:interventional_faithfulness_implications_v_structure}, $X_i\to X_k\leftarrow X_j$ is a v-structure in $\mathcal D$ if and only if, for every conditioning set $\mathbb{Z}\subseteq \mathbb{V}\setminus\{X_i,X_j\}$ containing $X_k$,
$$
{r_{X_iX_j\cdot \mathbb{Z}}}_{\mathbb{P}^1}
\neq
{r_{X_iX_j\cdot \mathbb{Z}}}_{\mathbb{P}^2}.
$$
Since $Sepset(X_i,X_j)$ contains the conditioning sets for which equality of regression coefficients was found, Lemma~\ref{lemma:interventional_faithfulness_implications_v_structure} implies that $X_i\to X_k\leftarrow X_j$ is a v-structure in $\mathcal D$ if and only if no separating set recorded in $Sepset(X_i,X_j)$ contains $X_k$. Therefore, the v-structure orientations performed by LDiffPC are sound.

Finally, after all v-structures have been correctly oriented, LDiffPC applies Meek's rules. By Lemma~\ref{lemma:meek}, successive applications of Meek's rules to the pattern of a DAG produce only correct orientations. Therefore every orientation added during this final phase is also present in the true difference DAG $\mathcal D$.

Combining the three steps, $\hat{\mathcal D}$ and $\mathcal D$ have the same skeleton, and every orientation present in $\hat{\mathcal D}$ is also present in $\mathcal D$.
\end{proof}

\end{document}